%% file: main.tex
\documentclass[review,hidelinks,onefignum,onetabnum]{siamart220329}

\usepackage{booktabs}
\usepackage{soul,placeins}

\input{shared_info}

\ifpdf
\hypersetup{
  pdftitle={LPAA model},
  pdfauthor={S. Brozak, S. Peralta, T. Phan, J. Nagy, and Y. Kuang}
}
\fi

\begin{document}

\maketitle

\begin{abstract}
Flour beetles (genus \textit{Tribolium}) have long been used as a model organism to understand population dynamics in ecological research. A rich and rigorous body of work has cemented flour beetles' place in the field of mathematical biology. One of the most interesting results using flour beetles is the induction of chaos in a laboratory beetle population, in which the well-established LPA (larvae-pupae-adult) model was used to inform the experimental factors which would lead to chaos. However, whether chaos is an intrinsic property of flour beetles remains an open question. Inspired by new experimental data, we extend the LPA model by stratifying the adult population into newly emerged and mature adults and considering cannibalism as a function of mature adults. We fit the model to longitudinal data of larvae, pupae, and adult beetle populations to demonstrate the model's ability to recapitulate the transient dynamics of flour beetles. We present local and global stability results for the trivial and positive steady states and explore bifurcations and limit cycles numerically. Our results suggest that while chaos is a possibility, it is a rare phenomenon within realistic ranges of the parameters obtained from our experiment, and is likely induced by environmental changes connected to media changes and population censusing.
\end{abstract}



\section{Introduction}
It is well known that limited resources will affect population growth, and that, on the other hand, organisms may affect their environment. Naturally, population fluctuations will be observed as this dance goes on---organisms will attempt to adapt to their environment while potentially modifying their environment and then responding to these changes. The potential for chaotic dynamics, characterized by aperiodic oscillations, was recognized in single-species ecological models by Robert May in 1974 \cite{may_1974}, spurring a hunt for chaos in natural populations.

In a landmark study, Costantino and colleagues induced chaotic dynamics in a laboratory population of red flour beetles \cite{costantino_1995,costantino_1997}. In this work, laboratory populations of \textit{Tribolium castaneum} (Family: Tenebrionidae) were manipulated in order to place their dynamics in the desired region of asymptotic behavior, such as convergence to a stable equilibrium or that of a limit cycle. This study spurred a decade-long exploration into the rich dynamics of the LPA (larvae-pupae-adult) model and the role of chaos in natural ecosystems. The ``Beetle Team", comprised of Jim Cushing, R. F. Costantino, Brian Dennis, Robert Desharnais, Shandelle Henson, Aaron King, and Jeffrey Edmunds, thoroughly characterized the dynamical behaviors of the LPA model and rigorously validated their theoretical results with experiments (see, for example, \cite{dennis_1995,edmunds_2003,king_2004,cushing_1998_saddle,henson_2002_basins}).

\textit{Tribolium}'s oscillatory population dynamics, as well as their adaptation to cohabitation with humans, made the insects an excellent candidate organism in the hunt for chaos. The humble flour beetle has long been used to study reproduction, population dynamics, evolution, genetics, and dispersal \cite{pointer_2021,park_1934}. The flour beetle's over 5,000-year history with infesting stored grains means the insects can easily adapt to their experimental settings \cite{andres_1931}. Numerous experiments by Thomas Park have solidified the place of \textit{Tribolium} in scientific literature \cite{park_1948,park_1954,park_1957}; his two-species competition experiments in particular spurred ecological and mathematical exploration into why the experiments did not end with a consistently dominant species \cite{park_1948,bartlett_1957,leslie_1962}. A thorough review of the contributions of flour beetles to ecology and biology may be found in \cite{pointer_2021}. \textit{Tribolium} have also been modeled mathematically prior to their most recent rise to fame; for example, competition between two \textit{Tribolium} species has been modeled using a system of ODEs, with Bartlett modeling competition from a predator-prey perspective rather than that of resource competition \cite{bartlett_1957}. Leslie modeled this competition using a stochastic two-stage model for each species, although finding the two-stage age-structure insufficient \cite{leslie_1962}. Renshaw studied the spatial distribution of flour beetles using a discrete ``stepping-stone" lattice model as well as analyzing a diffusion model \cite{renshaw_1980}.

Part of the mathematical interest in \textit{Tribolium} stems from nonlinear interactions between their life stages. Flour beetles are holometabolous (i.e., go through complete metamorphosis) and self-regulate their populations through cannibalism \cite{park_1934,brindley_1930,mertz_1970,pointer_2021}. It has been suggested that the periodic nature of \textit{Tribolium} populations is induced by cannibalistic behaviors \cite{pointer_2021} (in general, Veprauskas and Cushing showed that sufficiently intense cannibalism on juveniles prevents extinction in nutrient-poor environments \cite{veprauskas_2017}). Adult females lay between two to sixteen eggs per day \cite{sonleitner_1991,arnaud_2005}. These eggs may be cannibalized by larvae or adults; those that survive go through six larval instars \cite{park_1934}. After approximately fourteen to thirty days, larvae develop into pupae and remain in this state for seven to fourteen days \cite{park_1934,park_1935}. Pupae are immobile and have no protection, and may be consumed by adults and sometimes larvae \cite{benoit_1998}. As they emerge, adults are initially white to light brown in color and do not have a hardened exoskeleton (sclerotization); these individuals may be referred to as callow adults and will develop a sclerotized exterior within a few days \cite{park_1934}. Newly emerged adult females have approximately 20 times lower fecundity than their mature counterparts \cite{park_1934}. New adults may also be consumed by mature adults \cite{benoit_1998}. Development times depend on humidity as well as temperature, with optimal egg-to-adult development at 30$^\circ$C lasting 30 days \cite{howe_1960,park_1934,pointer_2021}. Faster development is associated with higher temperatures and higher relative humidity \cite{park_1934}. 

The standard LPA model is given by \cite{costantino_1995,cushing_1994}:
    \begin{align}\label{eqn:LPA}
        \begin{split}
            L(t+1) & = bA(t)e^{-c_1L(t) - c_2A(t)},\\
            P(t+1) &= \left(1-\mu_l\right)L(t),\\
            A(t+1) &= P(t)e^{-c_3A(t)} + \left(1-\mu_a\right)A(t).
        \end{split}
    \end{align}
The populations of larvae, pupae, and adults at time $t$ are denoted by $L(t)$, $P(t)$, and $A(t)$, respectively. Units of time are in two-week intervals, which roughly aligns with the development time of larvae \cite{costantino_1995,park_1934,park_1935}. Hence, $P(t)$ contains pupae as well as non-feeding larvae and callow adults \cite{costantino_1995}. The recruitment rate $b$ denotes the number of eggs that will hatch into larvae in the absence of cannibalism \cite{costantino_1995}. The parameter $0<\mu_i<1$ is the natural mortality probability for life stage $i\in\{l,a\}$, so that $(1-\mu_i)$ is the proportion of individuals in life stage $i$ that survive to the next stage \cite{costantino_1995}. The coefficients $c_1\geq0$, $c_2\geq0$, and $c_3\geq0$ represent the intensity of cannibalism of eggs by larvae and adults, respectively, as well as adults consuming pupae \cite{costantino_1995}. 

\begin{figure}
    \centering
    \includegraphics[trim={0 12cm 0 10cm},clip,width=0.75\textwidth]{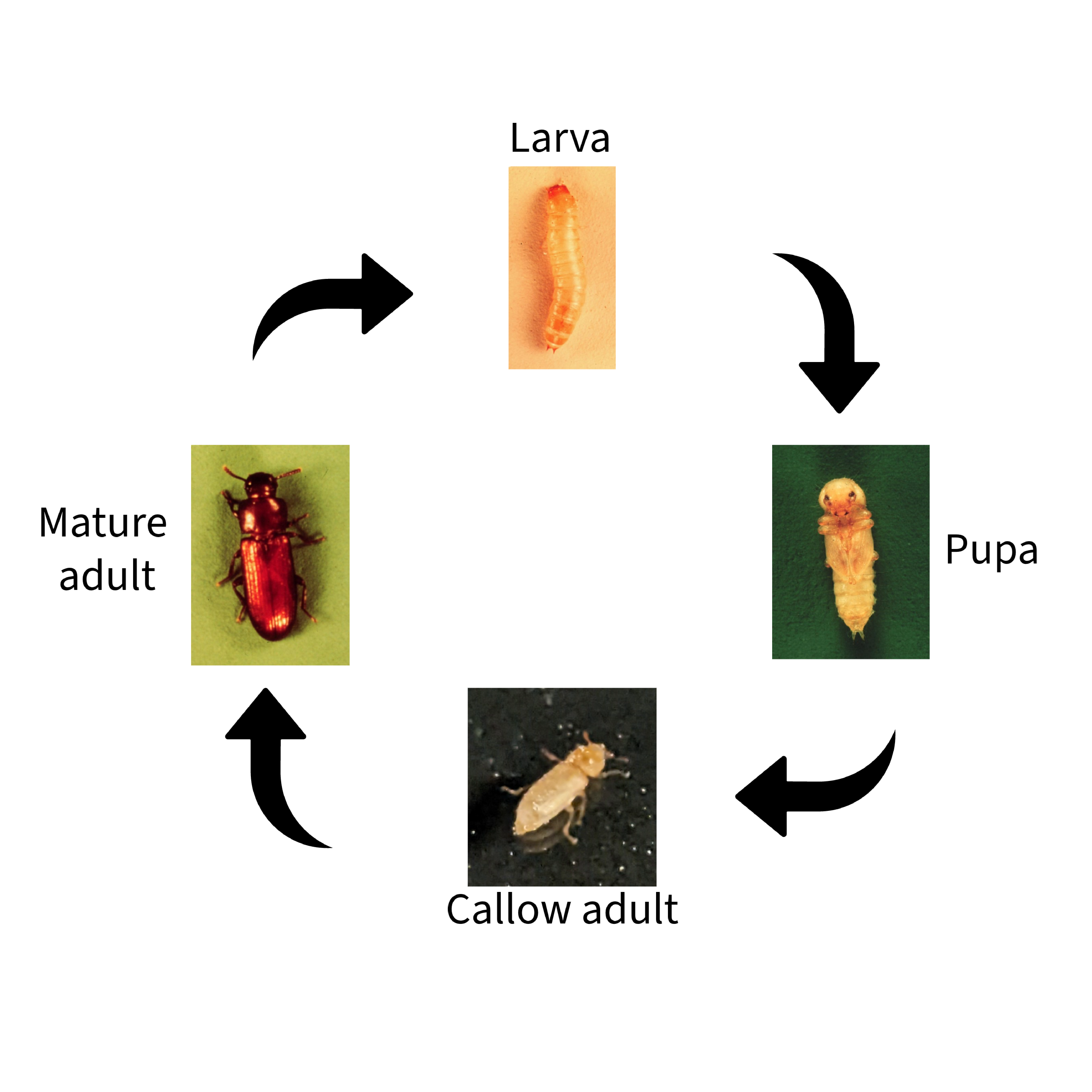}
    \caption{Life cycle of \textit{Tribolium} as depicted in the LPAA model \eqref{eqn:LPAA}. Images of larva, pupa, and mature adult taken from \cite{usda_ars}; image of callow adult provided by S. Alme.}
    \label{fig:lifecycle}
\end{figure}

The dynamics of the LPA model have been extensively characterized and validated with experimental data. In the absence of cannibalism, the model is linear and the population may go extinct or grow exponentially \cite{cushing_2001}. Cannibalism affects the asymptotic behavior of the system and induces complex nonlinear behaviors such as limit cycles and chaos. Cushing \cite{cushing_1994} showed the local stability of the extinction equilibrium as well as the stability of the positive steady state when the net reproductive threshold is near one. Furthermore, when the net reproductive threshold is greater than one, the system is uniformly persistent with respect to the extinction equilibrium. The model also has a global chaotic attractor which, for certain parameters, has an unstable saddle cycle of period eleven \cite{cushing_2001}. Kuang and Cushing \cite{kuang_1996}
derived thresholds for the global stability of the positive steady state in the absence of larval cannibalism on eggs. Desharnais and colleagues showed that small perturbations may have a significant effect on the amplitude of oscillations \cite{desharnais_2001}. In a special case of the model where adults live for two weeks, Cushing used synchronous orbits to study the existence of an invariant loop and cycle chains \cite{cushing_2003}.

The standard (deterministic) LPA model has been modified in several ways. Cushing and colleagues argue for the incorporation of demographic and environmental stochasticity \cite{cushing_2001,dennis_1995,dennis_1997,dennis_2001} and have also modified the model to account for a periodic-forcing environment \cite{costantino_1998,henson_1997,henson_2000}. For certain sets of parameters such as cannibalism rates and mortality, the LPA model predicted chaotic dynamics, which accurately reflected the data; this includes the prediction of invariant loops and equilibria \cite{dennis_1995,dennis_1997,dennis_2001,costantino_1995,costantino_1997}. A compilation of some mathematical results can be found in \cite{cushing_2004}.

Henson and Cushing proved the existence and stability of periodic cycles in the periodic LPA model, showing theoretically that a fluctuating environment results in increased population levels \cite{henson_1997}. These theoretical results were in line with experimental findings published by Jillson \cite{jillson_1980} and studied by Costantino et al. \cite{costantino_1998}. Henson et al. found in the periodic LPA model an unstable saddle cycle separated two stable 2-cycles, significantly affecting transient dynamics \cite{henson_1999}. Further, populations exhibiting oscillations in a constant environment may have more than one oscillatory final state in a periodic environment \cite{henson_2000,henson_2002_basins}.

Initially, we intended to study population dynamics in response to different environments using \textit{Tribolium confusum} (see \cref{app:exp}). Additional goals for this experiment were to decrease counting variability using protocols based on existing methods \cite{desharnais_1980,costantino_1995,costantino_1997}. \textit{T. confusum} has distinct physical life stage forms (i.e.: eggs, larvae, pupae, and adult) and experience behavioral shifts during the transition between each stage. We found it difficult to reproduce the protocols in previous works which grouped multiple life stages together, hence our desire to develop new protocols. 

After culturing flour beetles and recording their populations over several weeks, we attempted to parameterize the LPA model using our data; our results suggested that the dynamics exhibited in our laboratory may be better represented by a modified version of the LPA model. Hence, we present a four-dimensional discrete time map for the dynamics of \textit{Tribolium} beetles by accounting for the lower fecundity of newly emerged adults. We analyze the stability properties of the two steady states of the model and numerically study the bifurcations of the model. We were interested if an additional equation in the model would induce complexity, but found that our data supported the conclusion that chaos is not an inherent property of \textit{Tribolium}. Finally, we discuss differences in experimental methods and how they may contribute to model performance and generate hypotheses for further study.

\section{The LPAA model}\label{sec:LPAA}
The total beetle population is split into four mutually-exclusive compartments: larvae ($L$), pupae ($P$), newly emerged adults ($A_1$), and mature adults ($A_2$). Each time step spans two weeks, as done in previous work \cite{cushing_1998_nonlinear,costantino_1995}. Thus, the newly emerged adult compartment contains callow and newly sclerotized adults, and accounts for the significantly lower fecundity of newly emerged females \cite{park_1934}. As done by Park et al., we assume that it takes ten days for a newly-sclerotized adult to become sexually mature \cite{park_1968}. Hence, the two-week time step still holds for this formulation.  The discrete time model, which we will refer to as the LPAA model, is governed by the following difference equations:
    \begin{align}\label{eqn:LPAA}
        \begin{split}
            L(t+1) &= bA_2(t)e^{-c_1 A_2(t)},\\
            P(t+1) &= \left(1-\mu_l\right)L(t),\\
            A_1(t+1) &= \left(1-\mu_p\right)P(t),\\
            A_2(t+1) &= A_1(t)e^{-c_2 A_2(t)} + \left(1-\mu_a\right)A_2(t).
        \end{split}
    \end{align}
    
Similar to the LPA model \eqref{eqn:LPA}, cannibalistic interactions are modeled using a binomial distribution \cite{dennis_1995}. The strength of the cannibalistic interactions are described by the coefficients $c_1 \geq0$ (mature adults consuming larvae) and $c_2\geq0$ (mature adults consuming newly emerged adults). We assume that encounters between mature adults and an egg or newly emerged adult occur at random, with the probability that the egg or immature adult survives given by $(1-c_j)^{A_2(t)} \approx \exp{(-c_j A_2(t))}$ for $j = 1,2$ \cite{dennis_1995}. It is assumed that the only losses of eggs and immature adults are due to cannibalistic interactions with mature adults. Total or near-total survival of callow adults in the absence of cannibalism has been observed by Park et al. \cite{park_1968}.

While the LPAA model shares many parameters and functional forms with the LPA model, there are several key differences. Larval cannibalism on eggs and adult cannibalism on pupae are not included in the model; however, mature adult cannibalism on newly emerged adults is incorporated. While we extensively compared variations of the LPA and LPAA models which included various cannibalistic interactions (for example, adult cannibalism on pupae), the best fit was obtained when only adults partook in cannibalism. This result suggests the testable hypothesis that only these sources of cannibalism contribute significantly to the dynamics.

\section{Fitting to experimental data}
We investigated the influence of salts and large inorganic compounds at small concentrations on the population dynamics of \textit{T. confusum} over 20 weeks. This work is part of a larger ongoing study working to understand population dynamics in the context of multiple resource limitation from a stoichiometric perspective. Specifically, we aimed to understand how laboratory populations of \textit{Tribolium confusum} responded to changes to their environment in the form of varying nitrogen and phosphorus ratios. Details of our experimental protocols may be found in \cref{app:exp}. 

No significant qualitative differences were observed between experimental subgroups. We first fit the LPA model to our data, shown in the left panels of \cref{fig:Nub05} and \cref{fig:Pb05}. However, the estimated parameters were not biologically sensible, as shown in \cref{tab:param}. Adjusting the parameter constraints, the model failed to produce a good fit. This motivated a principled approach in developing an alternative, biologically-reasonable model that would better describe our results. After rigorously comparing many biologically-feasible models to our experiments, we found that the LPAA model best recapitulated our time-series data. 

Fittings to our experimental data are shown in \cref{fig:Nub05} and \cref{fig:Pb05}; we highlight these two cases since the other experimental subgroups had similar results\footnote{Data and other fittings are available at \url{https://github.com/sjbrozak/Tribolium-LPAA}.}. We minimize the weighted sum-squared residuals between the data and the one-step forecasts of the model. All datasets are weighted equally. The initial conditions of the juvenile stages were determined by the value of the data at the indicated time step. Since the adults have significantly longer lifespans than juveniles, we calculate the initial values of the adult populations using
    \begin{align*}
        A_1(j) &= data(j) - data(j-1),\\
        A_2(j) &= data(j-1)
    \end{align*}
where $data(j)$ indicates the number of adults at the $j$-th time step, $2\leq j\leq 9$.

One-step forecasts are shown in \cref{fig:Nub05}  and \cref{fig:Pb05} for two experimental groups which we feel are characteristic of the overall data. While the weighted sum of squared residuals are generally comparable between models (see \cref{tab:obj_vals}), estimated parameter values for the LPAA model lie in expected biological ranges \cite{sonleitner_1991,arnaud_2005}. The estimated values for the larval recruitment parameter $b$ indicate a poor fit for the LPA model. All experimental subgroups had biologically-reasonable LPAA parameters with the exception of the P 0.5\% bleached group, in which both the LPA and LPAA models performed poorly. QQ plots of the one-step residuals show approximately straight lines, indicating normally-distributed residuals and goodness of fit (see \cref{app:qq}). 

We hypothesize that, aside from using different \textit{Tribolium} species, different counting and sorting procedures also contribute to variation in model performance. For example, Desharnais and Liu combined large larvae and pupae to match the time steps of the LPA model \cite{desharnais_1987}, which we were unable to reproduce as thresholds for sorting were indeterminable.

\begin{table}[htbp]
\footnotesize
\caption{Parameter definitions. $^*$Larval mortality probability $\mu_l$ was estimated directly from the data.}\label{tab:param_desc}
\begin{center}
  \begin{tabular}{cll} \hline
        Parameter & Definition (unit) & Source\\\hline
        $b$ & Larval recruitment (larvae per adult) & Fitted \\
        $\mu_l$ & Proportion of larvae lost due to natural mortality & *\\
        $\mu_p$ & Proportion of pupae lost due to natural mortality & Fitted\\
        $\mu_a$ & Proportion of pupae lost due to natural mortality & Fitted\\
        $c_1$ & Cannibalism of mature adults on eggs (per mature adult) & Fitted\\
        $c_2$ & Cannibalism of mature adults on immature adults (per mature adult) & Fitted\\\hline
    \end{tabular}
\end{center}
\end{table}

\begin{table}
\footnotesize
\caption{Best fit parameter values. Median values are reported with minimums (maximums). References are listed for reported ranges found or assumed in the literature. It should be noted that, in the LPA model \eqref{eqn:LPA}, $c_1$ and $c_2$ correspond to cannibalism of eggs by larvae and adults. In the LPAA model \eqref{eqn:LPAA}, $c_1$ corresponds to cannibalism of eggs by mature adults and $c_2$ corresponds to cannibalism of newly emerged adults by mature adults. Larval mortality $\mu_l$ is estimated directly from the data, while all other parameters are fit.}
\begin{center}
{\begin{tabular}{llllll} \hline
 & \multicolumn{2}{c}{LPA} & & \multicolumn{2}{c}{LPAA} \\ \cmidrule{2-3} \cmidrule{5-6}
 Par. & Med. & (Min, Max) &  & Med. & (Min, Max) \\ \midrule
 $b$ & 20& (20, 20)&  & 6.4232& (4.2781, 20)\\
 $\mu_l$ & 0.6053 & (0.5253, 0.6739)&  & 0.6053 & (0.5253, 0.6739)\\
 $\mu_p$ &  &  &  & $2.64\text{e-}{12}$& ($1.21\text{e-}{12}$, $2.75\text{e-}{11}$)\\ 
 $\mu_a$ & 0.0842& (0.0353, 0.1039)& & 0.0358& (0, 0.0948)\\
 $c_1$ &0.0179& (0.0089, 0.0209)&  & 0.0099& (0.0066, 0.017)\\
 $c_2$ & 0.0003& (0, 0.0097)& & 0.0028& (0.0014, 0.0050)\\
 $c_3$ & $1.0760\text{e-}{13}$& ($1.8714\text{e-}{14}$, $7.922\text{e-}{5}$)& & &  \\
 \hline
\end{tabular}}
\label{tab:param}
\end{center}
\end{table}

\begin{table}[htbp]
\footnotesize
\caption{Objective values by experimental group. The only group with a poor parameterization (indicated by estimated parameters at the top of biological constraints) is the P 0.5\% bleached group.}\label{tab:obj_vals}
\begin{center}
  \begin{tabular}{lccc} \hline
        Group & LPA & LPAA & \% improvement \\\hline
        N 0.5\% unbleached & $\mathbf{3.68\times10^4}$& $3.98\times10^4$ & -8.15\% \\
        N 0.5\% bleached & $4.88\times10^4$ & $\mathbf{4.05\times10^4}$ & +16.91\%\\
        N 1\% unbleached & $\mathbf{3.50\times10^4}$ & $3.91\times10^4$& -11.71\%\\
        N 1\% bleached & $2.98\times10^4$ & $\mathbf{2.25\times10^4}$& +24.50\%\\
        \hline
        P 0.5\% unbleached & $3.31\times10^4$ & $\mathbf{2.97\times10^4}$& +10.2\%\\
        P 0.5\% bleached & $\mathbf{6.99\times10^4}$& $8.12\times10^4$& -16.17\%\\
        P 1\% unbleached & $\mathbf{5.66\times10^4}$& $6.09\times10^4$& -7.5972\%\\
        P 1\% bleached & $4.76\times10^4$& $\mathbf{4.11\times10^4}$ & +13.66\%\\      
        \hline
    \end{tabular}
\end{center}
\end{table}

\begin{figure}
    \centering
    \includegraphics[width=\textwidth]{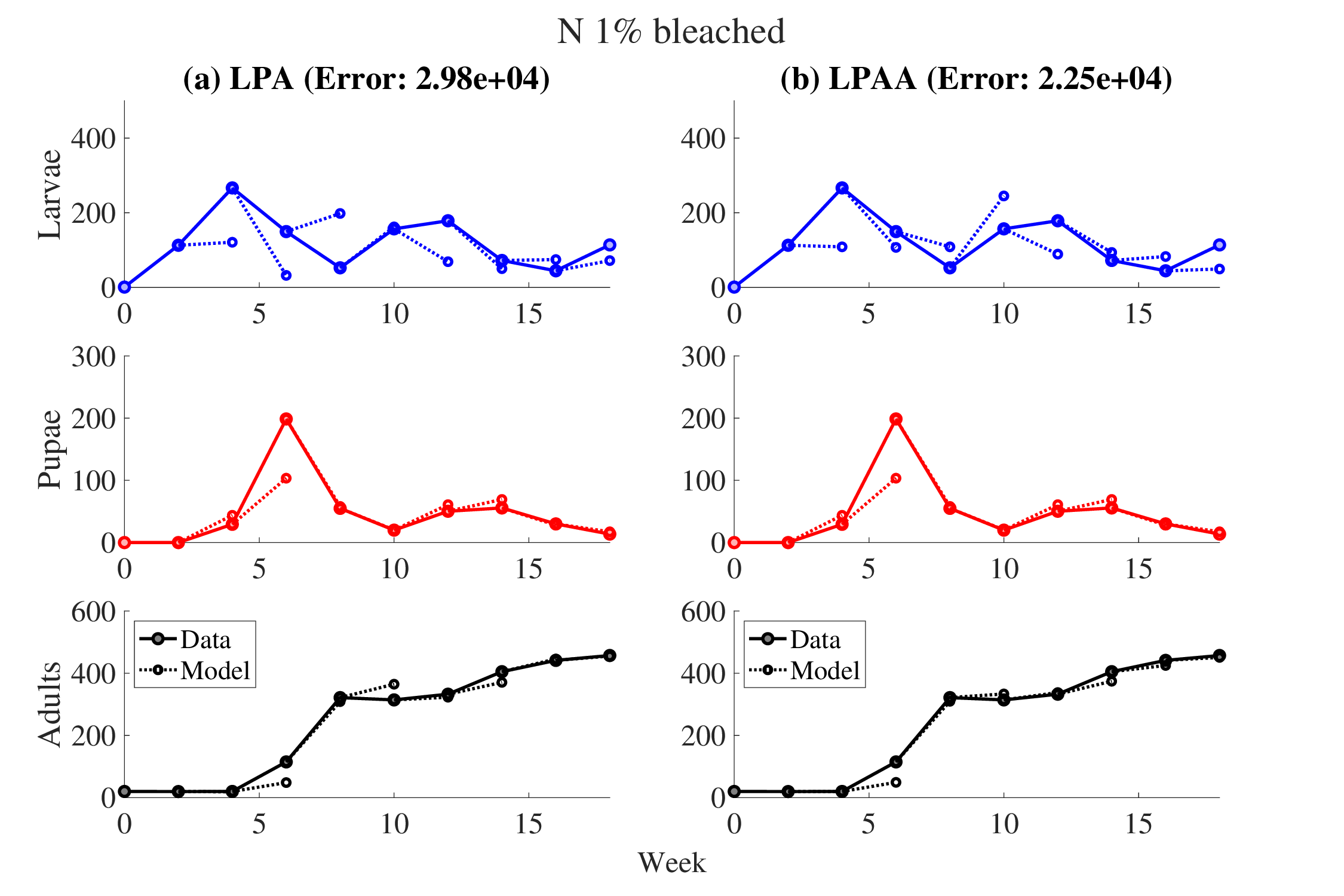}
    \caption{Left panels: LPA model fit to experimental group with 1\% nitrogen in bleached flour. Right panels: LPAA model fit to the same data. The weighted SSE is reported for each model.}
    \label{fig:Nub05}
\end{figure}

\begin{figure}
    \centering
    \includegraphics[width=\textwidth]{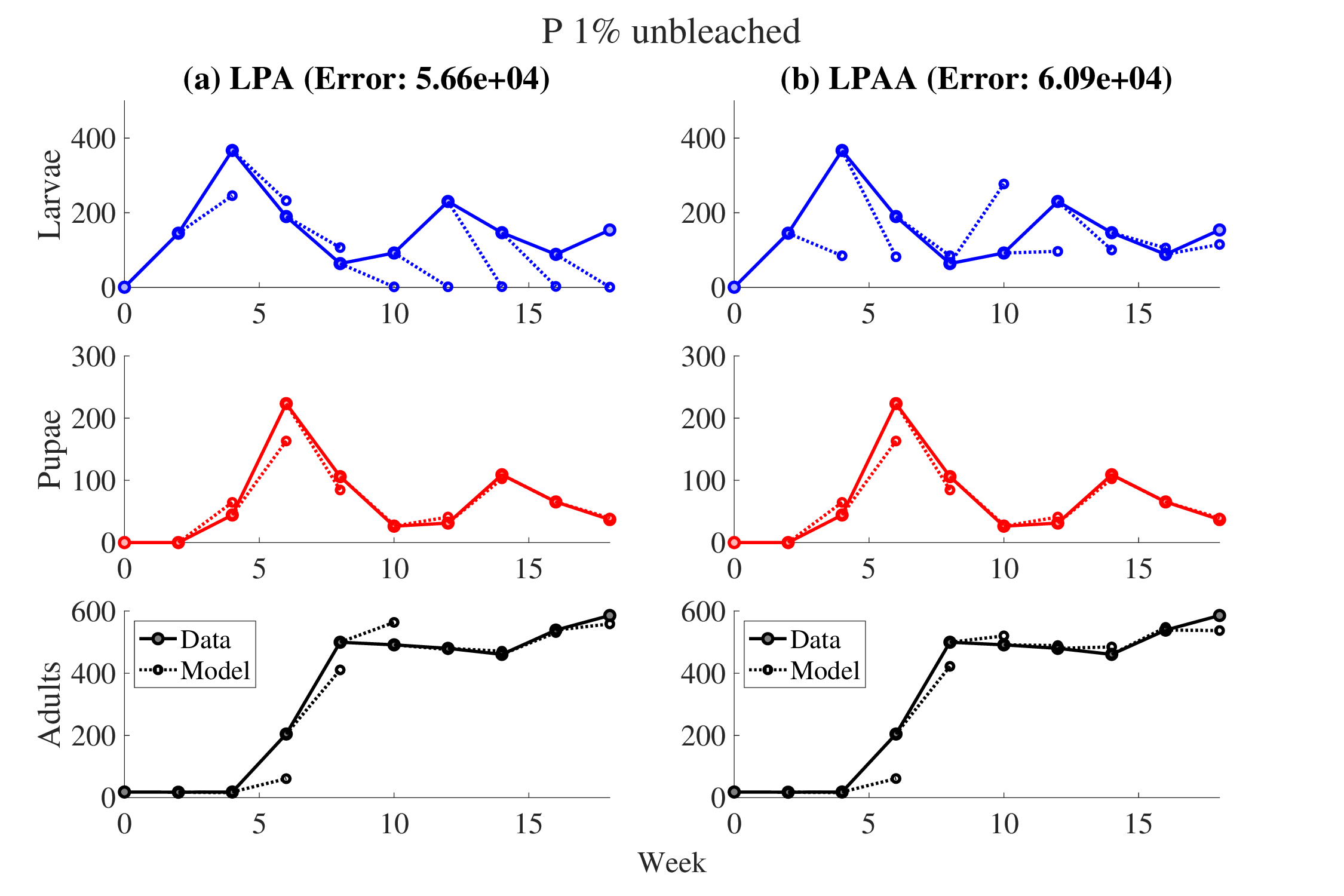}
    \caption{Left panels: LPA model fit to experimental group with 1\% phosphorus in unbleached flour. Right panels: LPAA model fit to the same data. The weighted SSE is reported for each model.}
    \label{fig:Pb05}
\end{figure}

\section{Basic model properties}
The LPAA model presented above is autonomous with the corresponding projection matrix
    \[
        \mathbf{P}(\mathbf{x}(t)) = \begin{bmatrix}
              0 & 0 & 0 & be^{-c_1 A_2(t)}\\
              1-\mu_l & 0 & 0 & 0\\
              0 & 1-\mu_p & 0 & 0\\
              0 & 0 & e^{-c_2 A_2(t)} & 1-\mu_a\\
    \end{bmatrix},
    \]
where $\mathbf{x}(t) = (L(t),\; P(t),\; A_1(t),\; A_2(t))^T$.
\begin{lemma}
    The LPAA model \eqref{eqn:LPAA} is positively-invariant for the region
        \[\Omega = \left\{\left(L,P,A_1,A_2\right)\in\mathbb{R}^4_+:L\leq \hat{L} , P\leq\hat{P}, A_1\leq\hat{A}_1\, A_2\leq \frac{\hat{A}_1}{\mu_a} + A_2(0) \right\}\]
    where $\hat{L} = \frac{b}{ec_1}$, $\hat{P} = \frac{b\left(1-\mu_l\right)}{ec_1}$, and $\hat{A}_1 = \frac{b\left(1-\mu_l\right)\left(1-\mu_p\right)}{ec_1}$.
 \end{lemma}
\begin{proof}
    Because the projection matrix is non-negative with non-negative initial conditions, solutions of the model exist uniquely and remain non-negative \cite{cushing_1998_cbms}. To show boundedness, observe that
        \[L(t) \leq \hat{L} := \max\left\{bxe^{-c_1x}:x>0\right\},\]
    where the maximum $\hat{L}=\frac{b}{ec_1}$ exists at $x=1/c_1$. Hence we may obtain the following upper bounds on pupae and immature adults,
        \[P(t)\leq \hat{P}:= \left(1-\mu_l\right)\hat{L}\]
    and
        \[A_1(t) \leq \hat{A}_1 := \left(1-\mu_p\right)\hat{P}.\]
    We now show that $A_2(t)$ is bounded by iterating the equation for mature adults. Observe
        \begin{align*}
            A_2(1) &\leq \hat{A}_1 + \left(1-\mu_a\right)A_2(0),\\
            A_2(2) &\leq \hat{A}_1 + \left(1-\mu_a\right)\hat{A}_1 + \left(1-\mu_a\right)^2A_2(0),\\
            \vdots\phantom{(3)} & \leq \phantom{asdfghj} \vdots\\
            A_2(n) & \leq \hat{A}_1 \sum_{i=0}^{n-1} \left(1-\mu_a\right)^i + \left(1-\mu_a\right)^nA_2(0).
        \end{align*}
    which implies
        \begin{align*}
            A_2(n) &\leq \hat{A}_1 \left(\frac{1- \left(1-\mu_a\right)^n}{1-\left(1-\mu_a\right)}\right) + \left(1-\mu_a\right)^nA_2(0)\\
            &\leq \frac{\hat{A}_1}{\mu_a} + A_2(0).
        \end{align*}
    Thus $A_2(t)$ is bounded, and subsequently all compartments are bounded for all time given non-negative initial data.
\end{proof}

The model \eqref{eqn:LPAA} attains two steady states: a trivial (extinction) steady state $\mathbf{E}^0 = (0,0,0,0)^T$ and a unique positive steady state given by $\mathbf{E}^* = (L^*, P^*,A_1^*,A_2^*)^T$, which satisfies 
    \begin{alignat*}{2}
       L^* &= bA_2^*e^{-c_1 A_2^*}, \qquad \qquad &&P^* =\left(1-\mu_l\right)L^*,\\
       A_1^* &= \left(1-\mu_p\right)P^*, && A_2^* = A_1^*e^{-c_2A_2^*} + \left(1-\mu_a\right)A_2^*.
    \end{alignat*}
We may obtain a closed form expression for $\mathbf{E}^*$ by substituting $L^*$ and $P^*$ into $A_1^*$, and we find
    \[A_1^* = bA_2^*\left(1-\mu_l\right)\left(1-\mu_p\right)e^{-c_1A_2^*}.\]
Assuming $A_2^*\neq0$ (lest the model stays at the extinction equilibrium) and substituting this new expression for $A_1^*$ into $A_2^*$,
    \[A_2^* = \frac{1}{c_1 + c_2} \ln{\left(\frac{b\left(1-\mu_l\right)\left(1-\mu_p\right)}{\mu_a} \right)}.\]
Clearly, the positive steady state exists uniquely when $\ln{\left(\frac{b\left(1-\mu_l\right)\left(1-\mu_p\right)}{\mu_a} \right)}>0$, and so we can define
    \[R_0 := \frac{b\left(1-\mu_l\right)\left(1-\mu_p\right)}{\mu_a}>1\]
as the condition necessary for existence of the positive steady state. Biologically, this threshold gives the average number of offspring of a single adult that survive to adulthood.

\section{Stability of the extinction steady state}
We present the conditions required for local and global stability of the extinction steady state.

\begin{theorem}
    The extinction equilibrium $\mathbf{E}^0$ is locally and globally asymptotically stable when $R_0<1$, and unstable when $R_0>1$.
\end{theorem}
    
\begin{proof}
Recall that the population projection matrix is given by
    \[ \mathbf{P}(\mathbf{x}(t)) = 
        \begin{bmatrix}
              0 & 0 & 0 & be^{-c_1 A_2(t)}\\
              1-\mu_l & 0 & 0 & 0\\
              0 & 1-\mu_p & 0 & 0\\
              0 & 0 & e^{-c_2 A_2(t)} & 1-\mu_a
              \end{bmatrix}.
    \]
The inherent projection matrix, which describes dynamics of small populations, is given by the Jacobian at the trivial equilibrium, or $\mathbf{P}(\mathbf{E}^0)$ \cite{cushing_1998_cbms}. Observe that a non-negative matrix $A$ is primitive and irreducible if and only if $A^m>0$ for some positive integer $m$ \cite[Ch. 13, section 5]{gantmacher_2000}. The matrix $(\mathbf{P}(\mathbf{E}^0))^m>0$ for $m=6$. We are now equipped to apply the Perron-Frobenius theorem, restated here from \cite[p. 182]{hofbauer_1998}.

\begin{theorem}[Perron-Frobenius]
    If $M$ is an $n\times n$ non-negative matrix, there exists a unique non-negative eigenvalue $\lambda$ which is dominant in the sense that $|\mu|\leq\lambda$ for all other eigenvalues $\mu$ of $M$. There exists right and left eigenvectors $u\geq0$ and $v\geq0$ such that $Mu=\lambda u$ and $vM=\lambda v$. If $M$ is primitive and irreducible, then $\lambda$ is simple and positive, $u$ and $v$ are unique and positive, and $|\mu|<\lambda$.
\end{theorem}

By the Perron-Frobenius theorem, $\mathbf{P}(\mathbf{E}^0)$ has a positive, algebraically simple, strictly dominant eigenvalue. Furthermore, this eigenvalue is less than one if and only if $R_0<1$, hence the extinction steady state is locally asymptotically stable when this condition holds  \cite[Theorem 1.1.3, p. 10]{cushing_1998_cbms}. Note that, element-wise,

    \[
    \begin{bmatrix}
              0 & 0 & 0 & be^{-c_1 A_2(t)}\\
              1-\mu_l & 0 & 0 & 0\\
              0 & 1-\mu_p & 0 & 0\\
              0 & 0 & e^{-c_2 A_2(t)} & 1-\mu_a
    \end{bmatrix} \leq
    \begin{bmatrix}
              0 & 0 & 0 & b\\
              1-\mu_l & 0 & 0 & 0\\
              0 & 1-\mu_p & 0 & 0\\
              0 & 0 & 1 & 1-\mu_a
              \end{bmatrix}.
    \]
Because the projection matrix is less than or equal to the inherent projection matrix for all $(L,P,A_1,A_2)\in \mathbb{R}^4_+$, a comparison argument shows that the extinction equilibrium is globally asymptotically stable when $R_0<1$ \cite{cushing_1998_cbms}.
\end{proof}

\section{Local stability of the positive steady state}
\begin{theorem}
    The positive steady state is locally-asymptotically stable when
    $$1<R_0<\min\left\{\exp\left(1 + \frac{c_2}{c_1}\right),\exp\left(\frac{1-\mu_a}{\mu_a}\left(1 + \frac{c_1}{c_2}\right)\right)\right\}.$$
\end{theorem}
\begin{proof}
    To derive conditions for local stability, we linearize around the positive steady state \cite{cushing_yicang_1994} and observe that the Jacobian of the system  \eqref{eqn:LPAA} takes the form
    \begin{equation*}
        J(L^*,P^*,A_1^*,A_2^*)= \begin{bmatrix}
            0 & 0 & 0 & L^*\left(\frac{1}{A_2^*}-c_1\right)\\
            1-\mu_l &0&0&0\\
            0& 1-\mu_p &0&0\\
            0 & 0 & \mu_a\frac{A_2^*}{A_1^*} & 1-\mu_a(1+c_2A_2^*)\\
        \end{bmatrix}.
    \end{equation*} 
This Jacobian will be non-negative, irreducible, and primitive (again the sixth power) given that
    $$L^*\left(\frac{1}{A^*_2}-c_1\right)>0 \quad \text{and}\quad 1-\mu_a(1+c_2A_2^*)>0.$$
The first inequality leads to $A_2^*<1/c_1$, while the latter implies
    $$A_2^*<\frac{1}{c_2}\cdot\frac{1-\mu_a}{\mu_a}.$$
Thus,
    $$R_0 < \min\left\{\exp\left(1 + \frac{c_2}{c_1}\right),\exp\left(\frac{1-\mu_a}{\mu_a}\left(1 + \frac{c_1}{c_2}\right)\right)\right\}.$$
It should be noted that these results are analogous to those of Cushing and Zhou for the LPA model \cite{cushing_yicang_1994}.
\end{proof}

\section{Global stability of the positive steady state}
We show global stability for the positive equilibrium following an argument similar to Kuang and Cushing \cite{kuang_1996}. Specifically, we aim to use the result proved by Hautus and Bolis \cite{emerson1979} and restated in Kuang and Cushing \cite[Theorem 1.1]{kuang_1996}. The spirit of this theorem relies on showing that the dynamical system, when constrained to a region $D$, is monotone with only one attracting steady state. 

We follow \cite{kuang_1996} and convert the LPAA model \eqref{eqn:LPAA} to a single discrete delay equation for $t\geq3$:
    \[A_2(t+1) = bA_2(t-3)\left(1-\mu_l\right)\left(1-\mu_p\right)e^{-c_1A_2(t-3)-c_2A_2(t)} + \left(1-\mu_l\right) A_2(t).\]
For convenience, we define
    \begin{alignat*}{2}
        \beta&:= b\left(1-\mu_l\right)\left(1-\mu_p\right), \quad &&\beta>0,\\
        x_t &:= A_2(t+3),  &&t\geq-3.
    \end{alignat*}
This yields
    \begin{equation}\label{eqn:xt}
        x_{t+1} = \left(1-\mu_a\right) x_t + \beta x_{t-3}e^{-c_1 x_{t-3} - c_2 x_t},  \quad t\geq0
    \end{equation}
with initial conditions $(L(0),P(0),A_1(0),A_2(0))^T$ converted to initial history
    \begin{align*}
        x_{-3} &= A_2(0) >0,\\
        x_{-2} &= A_2(1) = A_1(0)e^{-c_2A_2(0)} + \left(1-\mu_a\right)A_2(0),\\
        x_{-1} &= A_2(2) = P(0)\left(1-\mu_p\right)e^{-c_2A_2(1)} + \left(1-\mu_a\right)A_2(1),\\
        x_0 &= A_2(3) = L(0)\left(1-\mu_l\right)\left(1-\mu_p\right)e^{-c_2A_2(2)} + \left(1-\mu_a\right)A_2(2).
    \end{align*}
The difference-delay model \eqref{eqn:xt} attains a trivial steady state at $x^0 = 0$ and
    \[x^* = \frac{1}{c_1 + c_2}\ln{\left(\frac{\beta}{\mu_a} \right)},\]
which exists and is unique when $\beta/\mu_a > 1$. It should be noted that $\beta/\mu_a = R_0$ as we have defined previously; we will use the notation $R_0$ going forward. This condition is analogous to what was found by Kuang and Cushing \cite{kuang_1996} for the LPA model where instead they defined $\beta=b\left(1-\mu_l\right)$. Our results on global stability are summarized in the following theorem.

\begin{theorem}
    The positive steady state $E^*$ is globally asymptotically stable when
        \[1<R_0<\min{\left\{e, \frac{ec_1\left(1-\mu_a\right)}{c_2\mu_a} \right\}}.\]
\end{theorem}

\begin{proof}
We assume that $R_0>1$ so that this steady state exists. Before proceeding, we aim to bound the solutions of \eqref{eqn:xt} from above. For $t\geq0$, observe that
    \[x_{t+1} \leq \left(1-\mu_a\right) x_t + \beta x_{t-3}e^{-c_1 x_{t-3}}.\]
The function $xe^{-c_1 x}$ attains a maximum of $1/\left(ec_1\right)$ at $x = 1/c_1$, so then
    \[x_{t+1} \leq \left(1-\mu_a\right) x_t + \frac{\beta}{ec_1}.\]
Iterating through the map from initial condition $x_0$, we have
    \[x_{t+1} \leq \left(1-\mu_a\right)^{t+1} + \frac{\beta}{ec_1\mu_a}\left(1 - \left(1-\mu_a\right)^{t+1} \right).\]
Since $\left(1-\mu_a\right)<1$ by definition,
    \begin{equation}\label{eqn:globalsup}
        \limsup_{t\to\infty} x_t \leq \frac{\beta}{ec_1\mu_a}.
    \end{equation}
We require some additional preliminaries in order to apply the desired result \cite{emerson1979,kuang_1996}. Define
    \[F(x_t,x_{t-1},x_{t-2},x_{t-3}) = \left(1-\mu_a\right) x_t + \beta x_{t-3}e^{-c_1 x_{t-3} - c_2 x_t},  \quad t\geq0\]
and $G(u) = F(u,u,u,u) - u$ with $u>0$ and $u\neq x^*$. Because $\lim_{u\to\infty}G(u)<0$ and $G(u)=0$ if and only if $u=0$ or $u = x^*$, it is true that $\left(u - x^*\right)\left[F(u,u,u,u) - u\right]<0$ as desired.

To show that $F$ is increasing in each of its arguments, we rewrite our difference-delay equation to include the transition compartments. We differentiate to obtain
    \begin{align*}
        \frac{\partial F}{\partial x_t} &= \left(1-\mu_a\right) - \beta c_2 x_{t-3} e^{-c_1x_{t-3} - c_2 x_t},\\
        \frac{\partial F}{\partial x_{t-1}} &= 0,\\
        \frac{\partial F}{\partial x_{t-2}} &= 0,\\
        \frac{\partial F}{\partial x_{t-3}} &= \beta\left(1 - c_1 x_{t-3}\right)e^{-c_1x_{t-3} - c_2 x_t}.\\
    \end{align*}

Assume that $\beta<\min\left\{e\mu_a, \frac{\left(1-\mu_a\right) ec_1}{c_2}\right\}$. Then there exists $T>3$ such that for $t>T$, $x_{t-3}<1/c_1$. We consider the region $D = \left( 0,1/c_1\right)^4$. Observe that $x^*\in\left( 0,1/c_1\right)$ and $\partial F/ \partial x_{t-i}\geq0$ for $i = 1,2,3$. We now focus on $\partial F/\partial x_t$. When $x_{t-3}\geq0$,
    \[x_{t-3}e^{-c_1 x_{t-3} - c_2 x_t} \leq x_{t-3}e^{-c_1 x_{t-3}} \leq \frac{1}{ec_1}.\]
The above implies that, for $x_t,x_{t-3}\geq0$,
    \[\frac{\partial F}{\partial x_t} \geq \left(1-\mu_a\right) - \frac{\beta c_2}{ec_1}\geq0\]
by assumption. This shows that, when restricted to the region $D$, $F$ is strictly increasing in its arguments. Hence we may apply the theorem proven by \cite{emerson1979}. Recalling that we require $R_0=\beta/\mu_a>1$, we have shown that the positive steady state is globally stable when
    \[\mu_a<\beta<\min{\left\{e\mu_a, \frac{ec_1\left(1-\mu_a\right)}{c_2} \right\}}.\] It should be noted that, in the case $c_2$ is small, the interval for global stability reduces to $1<R_0<e$. 
\end{proof}

\cref{fig:positiveLAS}(a) highlights the stability regions for the extinction and positive steady states as a function of adult mortality ($\mu_a$) and the natural logarithm of larval recruitment in the absence of cannibalism ($\ln{b}$). Numerical simulations corroborate our analytical findings for stability (right panels of \cref{fig:positiveLAS}). \cref{fig:positiveLAS}(a) shows sustained limit cycles outside of our analytical stability regions.

\begin{figure}
    \centering
    \includegraphics[width=\textwidth]{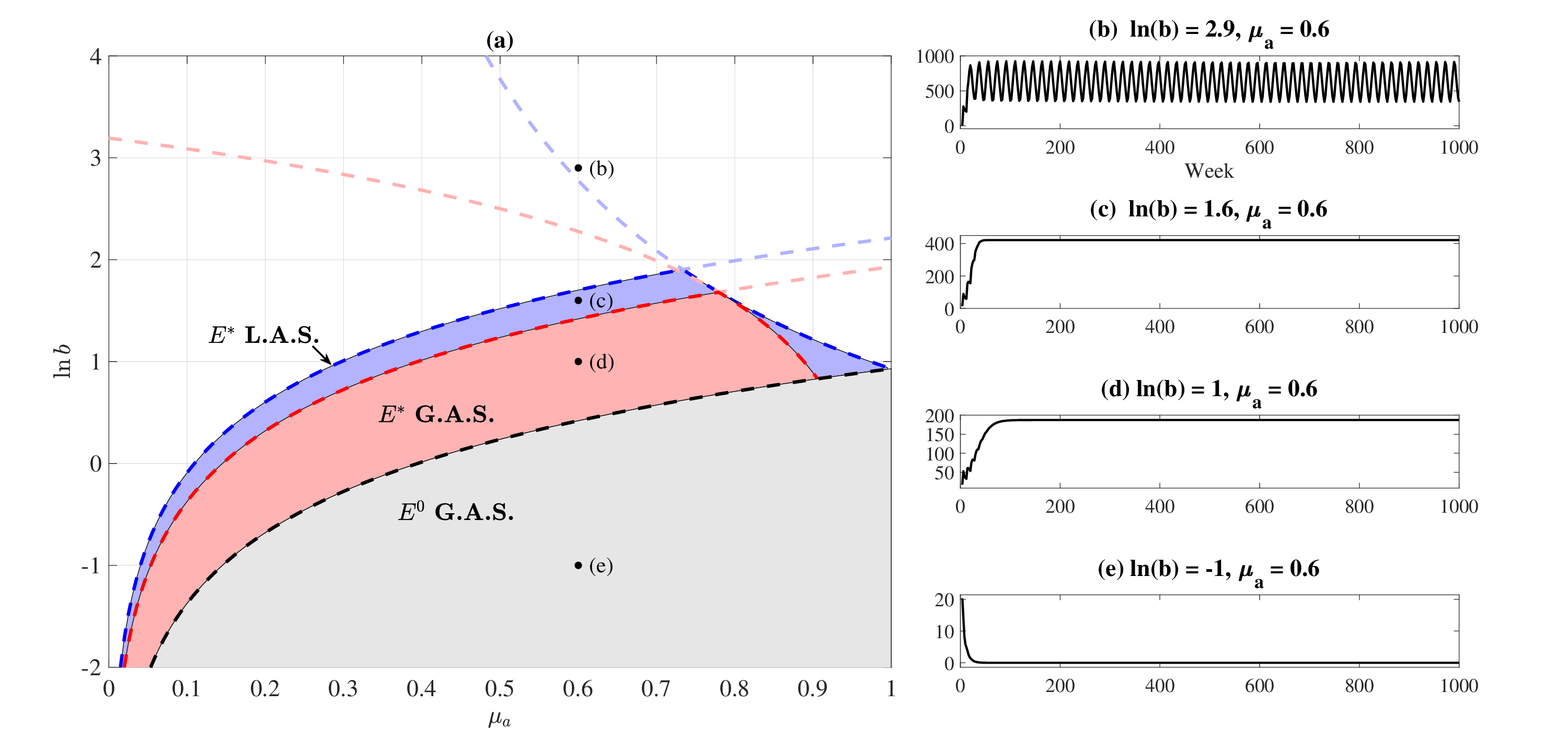}
    \caption{(a) Stability regions for steady states of the LPAA model \eqref{eqn:LPAA}, with the hypothesized upper bound for local stability in the special case when $c_2=0$. (b)--(e) show the total \textit{Tribolium} population for the indicated parameter values with $c_2 = 0$, $\mu_l = 0.6053$, and $\mu_p = 0$, similar to our fitted parameters listed in \cref{tab:param}.}
    \label{fig:positiveLAS}
\end{figure}

\section{Bifurcations}
We explore the bifurcation behavior of the LPA \eqref{eqn:LPA} and LPAA models \eqref{eqn:LPAA} numerically through simulation and Lyapunov exponent calculation to investigate if chaos is inherent in our laboratory populations. The Lyapunov exponent often acts as an indicator for chaotic behavior, although not sufficient \cite{cushing_2001}; specifically, strictly positive Lyapunov exponents are taken to denote chaos.  We follow the algorithm described in \cite[see appendix]{dennis_2001}, which we briefly restate here. In theory, the Lyapunov exponent $\lambda$ may be computed as
    \[\lambda = \lim_{t\to\infty} \frac{1}{t}\ln{\left||\mathbf{J}_t \mathbf{J}_{t-1}\dots \mathbf{J}_1\right||}\]
where $\mathbf{J}_t$ denotes the Jacobian of the system evaluated at time $t$. Here,
    \[\mathbf{J}_t = \begin{bmatrix}
              0 & 0 & 0 & b\left(1-c_1A_2(t)\right)e^{-c_1 A_2(t)}\\
              1-\mu_l & 0 & 0 & 0\\
              0 & 1-\mu_p & 0 & 0\\
              0 & 0 & e^{-c_2 A_2(t)} & -c_2A_1(t)e^{-c_2 A_2(t)}+(1-\mu_a)
              \end{bmatrix}.\]
However, the matrix multiplication may be numerically unstable. Following \cite{dennis_2001}, we rescale the matrix multiplication. A sequence of scalars $s_t$ is chosen to be $s_t = \left|| \mathbf{J}_t \mathbf{S}_{t-1}\mathbf{S}_{t-2}\dots\mathbf{S}_t\right||$ with $\mathbf{S}_t = \mathbf{J}_t/s_t$. The scalars are initialized with $s_1 = \left|| \mathbf{J}_1\right||$ and $\mathbf{S}_1 = \mathbf{J}_1/s_1$. Then, the Lyapunov exponent may be computed as
    \[\lambda = \frac{1}{t}\sum_{i=1}^t \ln{\left(s_i\right)}.\]
Further details on the calculations involved may be found in \cite[Appendix]{dennis_2001}. Our procedure for the bifurcation and Lyapunov exponent diagrams is as follows:

\begin{enumerate}
    \item Set the value for the chosen bifurcation parameter.
    \item Simulate the model \eqref{eqn:LPAA} for $50,000$ time steps in order to remove transients.
    \item For the bifurcation diagram, plot the last 100 iterations of the simulation for the given parameter value.
    \item For the Lyapunov exponent diagram, compute the Jacobian of the model evaluated at the current time step $\mathbf{J}_t$ and follow the computation above as stated in \cite[Appendix]{dennis_2001}.
    \item Update the bifurcation parameter, continue steps 2-4.
\end{enumerate}

\cref{fig:b} shows the dynamics of the model with the recruitment rate $b$ as the bifurcation parameter using parameter values obtained from fitting to the experimental data (see \cref{tab:param}). Both the LPA model (\cref{fig:b}(a)) and the LPAA model (\cref{fig:b}(b)) go to the positive steady state, with the Lyapunov exponents remaining negative and thus indicating non-chaotic behavior.

When the strength of cannibalism of eggs by adults is varied, we see that both models again go to a steady state around our experimental parameterization (\cref{fig:cegg}). No levels of egg cannibalism here show clear indication of chaotic behaviors, with Lyapunov exponents remaining negative.

Costantino and colleagues found chaotic dynamics, again indicated by positive Lyapunov exponent, when the intensity of cannibalism on pupae $c_3$ is between 0.1 and 0.5 in their experimental parameterization \cite{costantino_1997}. Because we do not incorporate cannibalism on pupae but instead on newly emerged adults, we instead vary $c_2$ for comparison. \cref{fig:cadult} indicates no chaotic dynamics in our experiments as the cannibalism intensity on new adults increases.

The clearest chaotic dynamics are shown when $\mu_a$, or adult mortality, is varied, as shown in \cref{fig:mua}. The LPA model indicates potentially chaotic dynamics (\cref{fig:mua}(a)), while the LPAA model goes to the positive steady state. In Costantino et al. \cite{costantino_1997}, the chaotic cloud was observed for $\mu_a>0.96$ based on their experimental data, while we do not observe chaotic dynamics in any region of $\mu_a$ for our model.

Biologically, these figures could indicate that chaos is not an inherent characteristic of \textit{Tribolium confusum}, but rather a response to environmental changes. While this conclusion is not surprising, our use of a different experimental setup and model strengthen the statement.

\begin{figure}
    \centering
    \includegraphics[width=\textwidth]{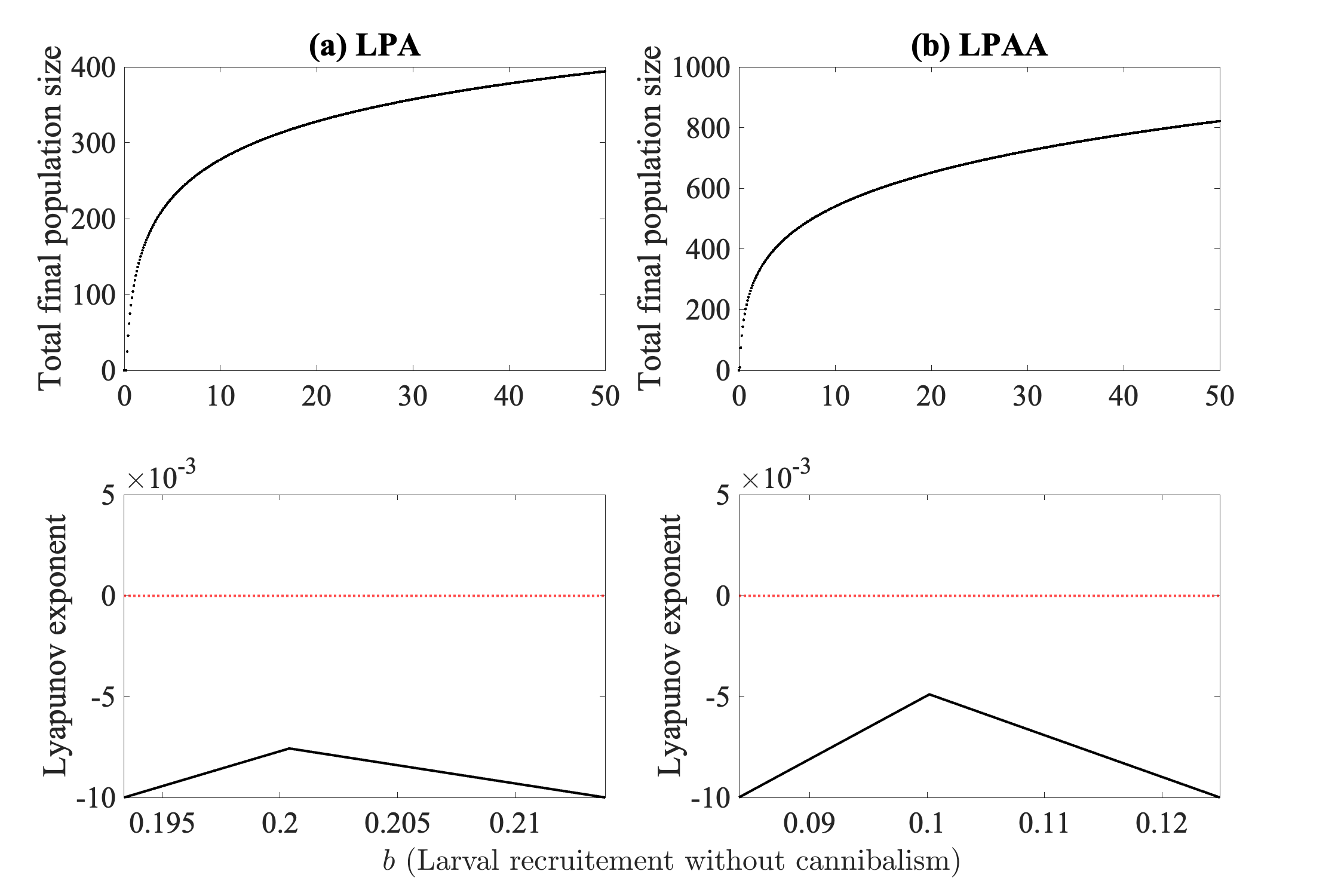}
    \caption{Bifurcation and Lyapunov exponent diagrams for (a) the LPA model and (b) the LPAA model where  larval recruitment in the absence of cannibalism is used as the bifurcation parameter. Median values obtained from fitting are used for all other parameters and are listed in \cref{tab:param}.}
    \label{fig:b}
\end{figure}

\begin{figure}
    \centering
    \includegraphics[width=\textwidth]{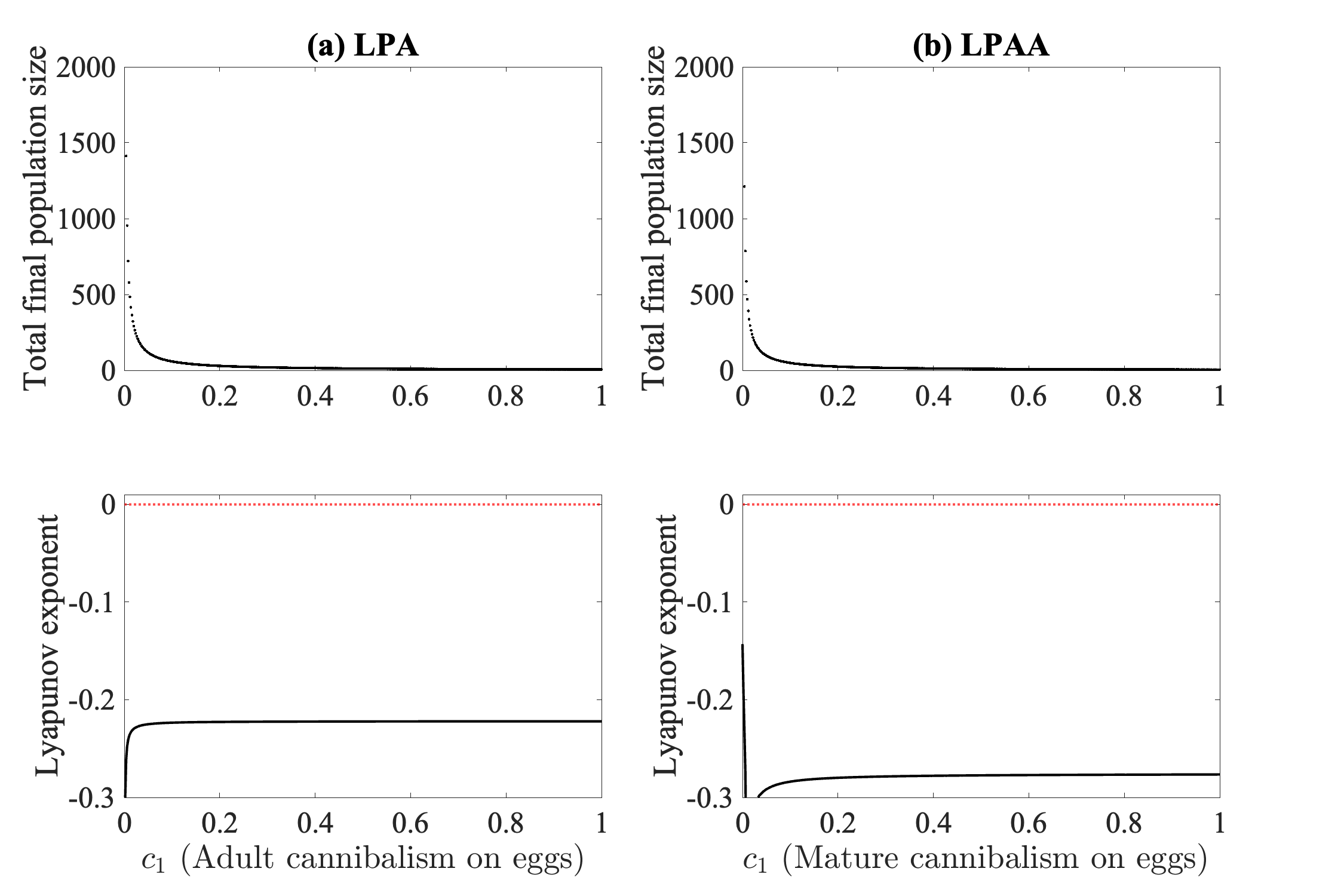}
    \caption{Bifurcation and Lyapunov exponent diagrams for (a) the LPA model and (b) the LPAA model where the cannibalism of eggs by adults is used as the bifurcation parameter. Median values obtained from fitting are used for all other parameters and are listed in \cref{tab:param}.}
    \label{fig:cegg}
\end{figure}

\begin{figure}
    \centering
    \includegraphics[width=\textwidth]{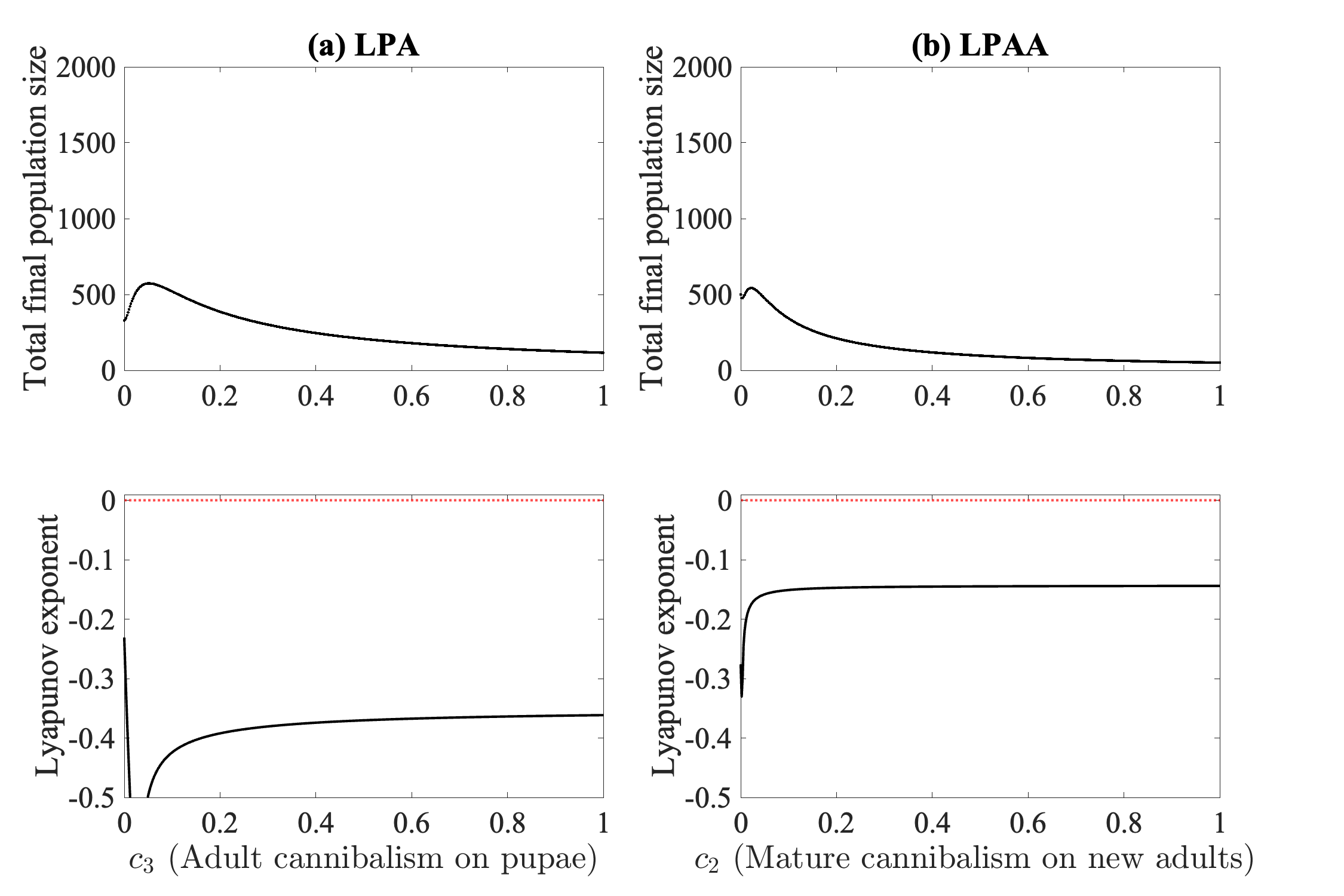}
    \caption{Bifurcation and Lyapunov exponent diagrams for (a) the LPA model and (b) the LPAA model where the cannibalism of the new adult cohort is used as the bifurcation parameter. Median values obtained from fitting are used for all other parameters and are listed in \cref{tab:param}.}
    \label{fig:cadult}
\end{figure}

\begin{figure}
    \centering
    \includegraphics[width=\textwidth]{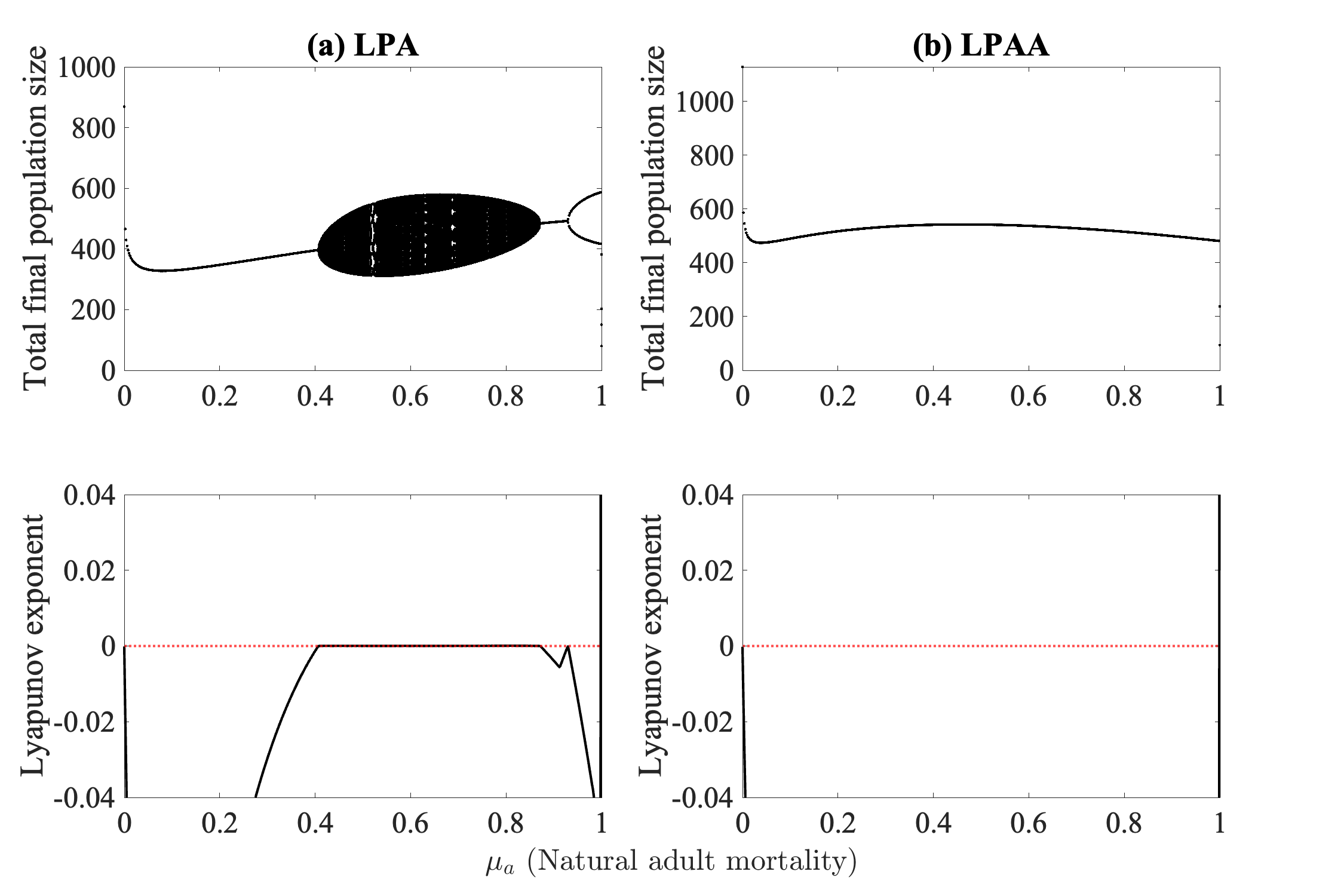}
    \caption{Bifurcation and Lyapunov exponent diagrams for (a) the LPA model and (b) the LPAA model where the proportion of adult lost due to natural causes is used as the bifurcation parameter. Median values obtained from fitting are used for all other parameters and are listed in \cref{tab:param}.}
    \label{fig:mua}
\end{figure}

\section{Discussion}
This work has contributed to the study of \textit{Tribolium} by extending the famous LPA model through the stratification of adults into newly emerged and reproductively mature. Although a four-dimensional matrix model, it was mathematically tractable. Numerical experiments support previous findings that chaos is not a natural characteristic of \textit{Tribolium} beetles and must be induced. Given their status as a grain pest, chaos may be observed as a result of pest management strategies pushing parameter values to chaotic regions (for example, sufficiently increasing adult mortality).

Initially, we cultured several \textit{Tribolium confusum} populations in order to understand population responses to varying environments. We found that the LPA model did not adequately describe our experiments, and a modified version of the LPA model, which we call the LPAA model, better recapitulated our time-series data. We added an additional equation for newly emerged adults motivated by their reduced fecundity. Our resulting model suggested that newly emerged adults may play an important role in flour beetle dynamics.

We analyzed the stability of the steady states of the LPAA model. The global stability results for the extinction and positive steady states were analogous to that of the LPA model. Bifurcations of the LPAA model were explored numerically, showing a distinct lack of chaos in regions in which the LPA model exhibited chaos. Analytic conditions for the stability of limit cycles is unknown. The parameterization obtained from fitting suggests that chaos is not inherent to our flour beetle population and must be induced.

We hypothesize the differences between the LPA and LPAA models are due to using different species of \textit{Tribolium} as well as different experimental setups. Aside from differing initial experimental goals, abiotic factors such as temperature and relative humidity may have a role. In addition, some data collected by the original Beetle Team categorized small larvae into the``larval stage" and large larvae and pupae (and potentially callow adults) into the ``pupal stage" in order to match the time discretization of the LPA model (for example, \cite{desharnais_1980} describes this stratification, but does not provide cutoffs for differentiating small and large larvae). As well, it is unclear exactly how media changes impact the population. Flour beetles condition the media via a buildup of feces, pheromones, and other excretions \cite{flinn_2012}. The conditioning of the flour has been shown to affect cannibalism and oviposition \cite{flinn_2012,sonleitner_1991}. The media changes  effectively ``reset" the flour and potentially influence the population dynamics. Media changes and censusing the population results in the loss of eggs with downstream effects on the later life stages, with the latter issue being particularly difficult to circumvent. While Costantino and colleagues changed media with each census (and thus every two weeks) \cite{costantino_1995}, media changes took place every eight weeks in our experiments. Hence, we hypothesize that media changes as well as censusing contribute to the oscillations, and potentially the chaos, observed in \textit{Tribolium}. 

This model is a direct extension of the groundbreaking work done by Cushing, Costantino, Dennis, Desharnais, Henson, and others; they have surely cemented the place of the LPA model and \textit{Tribolium} in mathematical history. Flour beetles provide an excellent source of interesting nonlinear interactions between life stages and are amenable to experimentation. These insects have a long history with humans and long future with scientists as well.

\FloatBarrier
\appendix
\section{Experimental methods}\label{app:exp}

Over the course of 20 weeks, we investigated the influence of salts and large inorganic compounds at small concentrations on the population dynamics of \textit{Tribolium confusum}. Large molecules like organophosphates are known to cause desiccation by blocking quinone glands, and salts tend to inhibit water retention \cite{lemon_1966}.

All animals used in this study were obtained from a stock colony of \textit{T.\ confusum} maintained at Scottsdale Community College. The stock population was founded 15 years prior to the start of this study from a small colony (approximately 100 beetles) obtained from Ward's Science in Rochester, New York and continuously maintained using standard culturing procedures for this species \cite{park_1934}. Beetles were reared on unbromated whole wheat flour and baker's yeast media (95:5 by mass) in 3.3~L food-grade containers modified to allow air flow. The stock population was continuously incubated in the dark at approximately $28\pm 2^\circ$C. The stock population experienced a bottleneck during the COVID pandemic with population density in 2021 reaching perhaps 10\% of its maximum. Robust populations were reestablished following the bottleneck and provided beetles for these experiments.

Sample colonies were established at time zero with 20 sclerotized adults cultured in 25 grams of flour with 5\% yeast media and held in 118~mL finely vented food storage containers. The study design comprised of two independent experiments testing effects of variations in media on population dynamics. Both experiments were fully crossed, two-factor, balanced designs with $n = 3$ sample cultures in each treatment. In both experiments, flour type---bleached and unbleached---was one factor. The other factor tested lethality of particular salts. In experiment one, the second factor was dipotassium phosphate (K$_2$PO$_4$), with levels 0.5\% and 1\% by mass added to the media. In experiment two, concentration of sodium nitrate (NaNO$_3$) was the second factor, with the same levels.

Populations were censused every two weeks with media changes every eight weeks. Adults, pupae, and larger larvae were sifted out of the media with a fine flour sifter. Larvae were removed promptly from the sifter and remaining individuals were to be sorted by length through a threshold of approximately 2.75 to 3.00 mm. If larvae exceed the threshold, they were considered large enough to count. Pupae were manually separated and counted by hand. Adults are sorted into living and dead categories, photographed, and recorded with ImageJ software \cite{schneider_2012}. Callow (unsclerotized) adults were counted separately, but the final adult count includes both callow and sclerotized adults. Deceased callow adults were not observed in this experiment. Eggs were not counted because they are laid freely through the media and are white or colorless and therefore blend into the media, and hence unsortable manually \cite{park_1934}. During media changes, half of the media is removed by mass, excluding live adults, pupae, and larvae. Media is then replenished with 12.50 grams of new media mixtures based on trial requirements. Careful measures were taken to ensure as many live larvae were removed from the media prior to recording the media mass for removal. The number of dead adults were documented and included in the discarded media.

\newpage 
\FloatBarrier
\section{QQ plots of residuals}\label{app:qq}
\begin{figure}[h!]
    \centering
    \includegraphics[width=0.8\textwidth]{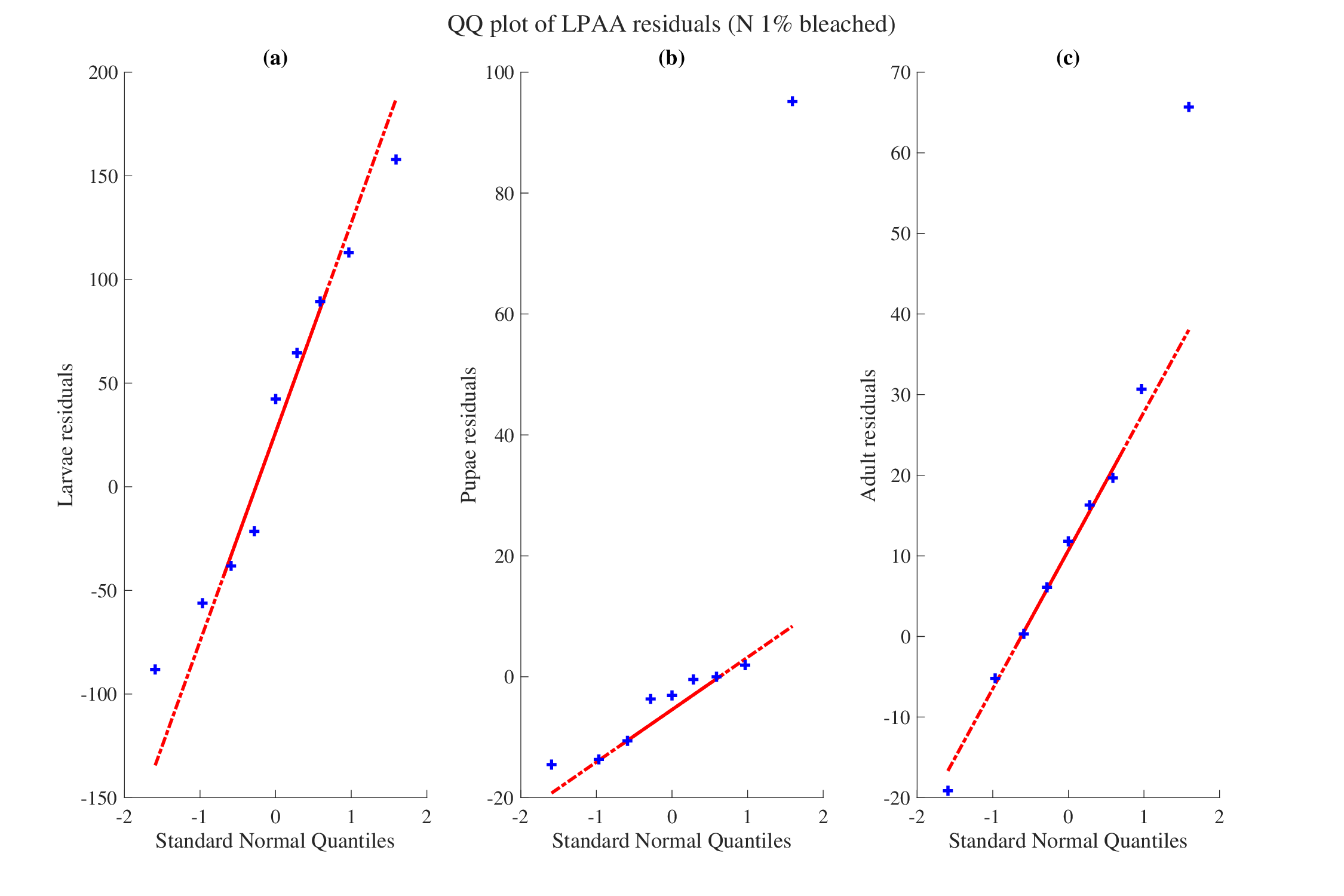}
    \caption{QQ plots of the residuals of the one-step forecasts of the LPAA model for the N 1\% bleached group.}
    \label{fig:qq_Nb1}
\end{figure}

\begin{figure}[h!]
    \centering
    \includegraphics[width=0.8\textwidth]{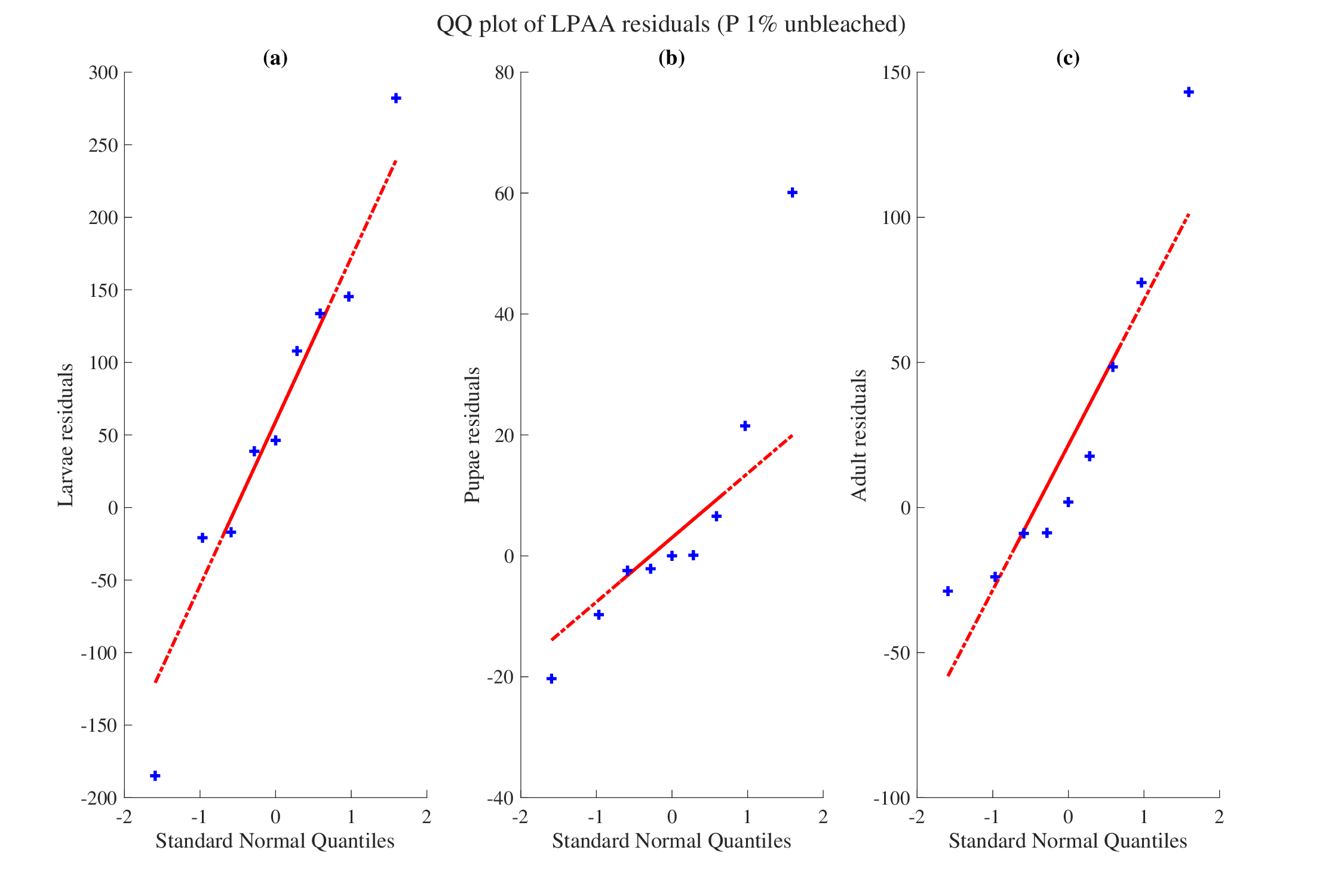}
    \caption{QQ plots of the residuals of the one-step forecasts of the LPAA model for the P 1\% unbleached group.}
    \label{fig:qq_Pub1}
\end{figure}
\FloatBarrier
\section*{Data availability}
All code and data used in this study are available on Github at the following repository: \url{https://github.com/sjbrozak/Tribolium-LPAA}.

\section*{Acknowledgments}
We extend our gratitude and appreciation of the hard work in the lab by Scottsdale Community College's own ``Beetle Team": Amanda Adams, Tatum Shepherd, and Michelle Hosking. We would also like to thank the editors and two reviewers whose feedback strengthened this work; in particular, we thank the first reviewer for their suggestions which sharpened the local stability proofs, and the second reviewer for their suggestions which improved the fittings and explanations of our model.

\bibliographystyle{siamplain}
\bibliography{ref}
\end{document}

%% file: shared_info.tex
\usepackage{lipsum}
\usepackage{amsfonts}
\usepackage{graphicx}
\usepackage{epstopdf}
\usepackage{algorithmic}
\ifpdf
  \DeclareGraphicsExtensions{.eps,.pdf,.png,.jpg}
\else
  \DeclareGraphicsExtensions{.eps}
\fi

\newsiamremark{remark}{Remark}
\newsiamremark{hypothesis}{Hypothesis}
\crefname{hypothesis}{Hypothesis}{Hypotheses}
\newsiamthm{claim}{Claim}

\headers{LPAA model}{S. Brozak, S. Peralta, T. Phan, J. Nagy, and Y. Kuang}

\title{Dynamics of an LPAA model for \textit{Tribolium} growth: Insights into population chaos\thanks{Initial submission 23 January 2024. Revisions submitted 4 June 2024. Accepted for publication in the SIAM Journal of Applied Mathematics 23 July 2024.
\funding{This work was supported in part by the National Institutes of Health (NIH) under Grant 5R01GM131405 (to YK and JN), the National Science Foundation (NSF) Rules of Life program DEB-1930728 (to SJB, TP, and YK),
eMB program DMS-2325146 (to SB, YK, and JN), and Los Alamos National Laboratory LDRD 20220791PRD2 (to TP).}}}

\author{Samantha J. Brozak\thanks{School of Mathematics and Statistical Sciences, Arizona State University, Tempe, AZ 
  (\email{sbrozak@asu.edu}, \email{kuang@asu.edu}, \email{John.Nagy@asu.edu}).}
\and Sophia Peralta\thanks{Department of Life Sciences, Scottsdale Community College, Scottsdale, AZ (\email{sophia.m.alme@gmail.com}).}
\and Tin Phan\thanks{Theoretical Biology and Biophysics, Los Alamos National Laboratory, Los Alamos, NM (\email{ttphan@lanl.gov}).}
\and John D. Nagy\footnotemark[2] \footnotemark[3]
\and Yang Kuang\footnotemark[2]}

\usepackage{amsopn}
